\documentclass[12pt,reqno]{amsart}

\usepackage{amsmath,amsfonts,amsbsy,amsthm}
\usepackage{eucal}
\usepackage{dsfont}

\usepackage{enumerate}
\usepackage{enumitem}
\setlist{leftmargin=*}

\numberwithin{equation}{section}

\makeatletter
\newtheoremstyle{corsivo}
   {\medskipamount}{\medskipamount}%
   {\itshape}{}%
   {\bfseries}{}%
   { }
   {\thmname{#1}\thmnumber{\@ifnotempty{#1}{ }\@upn{#2}}%
    \thmnote{ {\bfseries(#3)}}.}%
\makeatother

\theoremstyle{corsivo}
\newtheorem{thm}{Theorem}[section]

\newtheorem{prop}[thm]{Proposition}

\makeatletter
\newtheoremstyle{dritto}
   {\medskipamount}{\medskipamount}%
   {\rmfamily}{}%
   {\bfseries}{}%
   { }
   {\thmname{#1}\thmnumber{\@ifnotempty{#1}{ }\@upn{#2}}%
    \thmnote{ {\bfseries(#3)}}.}%
\makeatother

\theoremstyle{dritto}
\newtheorem{dfn}[thm]{Definition}
\newtheorem{rmk}[thm]{Remark}

\newtheorem{assumption}[thm]{Assumption}

\newcommand{\sub}[1]{_{\mathrm{#1}}}

\newcommand{\Id}{\mathds{1}}  

\newcommand{\iu}{\mathrm{i}}   

\newcommand{\N}{\mathbb{N}}
\newcommand{\Z}{\mathbb{Z}}

\newcommand{\R}{\mathbb{R}}
\newcommand{\C}{\mathbb{C}}

\newcommand{\Hi}{\mathcal{H}}

\newcommand{\Dom}{\mathcal{D}}

\DeclareMathOperator{\re}{Re} \DeclareMathOperator{\im}{Im}

\DeclareMathOperator{\expo}{exp}

\newcommand{\ie}{{\sl i.\,e.\ }}   

\newcommand{\virg}[1]{``#1''}

\newcommand{\X}{\mathfrak{X}}
\newcommand{\K}{\mathfrak{P}}
\newcommand{\Dd}{\mathcal{C}}
\newcommand{\aU}{\textbf{a}}
\newcommand{\aD}{\textbf{b}}
\newcommand{\Gabor}{\mathcal{G}}
\newcommand{\x}{{\bf{x}}}
\newcommand{\OC}{G}
\newcommand{\KM}{K}
\newcommand{\Hol}{\text{Hol}}

\usepackage{color}

\let\oldfootnote\footnote
\renewcommand{\footnote}[1]{\oldfootnote{\  #1}}

\title[Symmetry and localization for magnetic Schr\"{o}dinger operators]
{ Symmetry and localization\\[2mm] for magnetic Schr\"{o}dinger operators: \\[2mm]
Landau levels, Gabor frames and all that}

\author{Massimo Moscolari and Gianluca Panati}

\date{January 30, 2019 - Final version, to appear in {\it Acta Applicandae Mathematicae} (2019).}

\usepackage{hyperref}

\begin{document}

\maketitle

\begin{abstract}
We investigate the relation between broken time-reversal symmetry and localization of the electronic states, in the explicitly tractable case of the Landau model.

We first review, for the reader's convenience, the symmetries of the Landau Hamiltonian and the relation of the latter with the Segal-Bargmann representation of Quantum Mechanics. We then study the localization properties of the Landau eigenstates by applying an abstract version of the Balian-Low Theorem to the operators corresponding to the coordinates of the centre of the cyclotron orbit in the classical theory. Our proof of the Balian-Low Theorem, although based on Battle's main argument, has the advantage of being representation-independent. \\


\end{abstract}

\tableofcontents

\section{Introduction}

Symmetries play a prominent role in our understanding of the physical world. 
Solid State Physics is not exceptional, and researchers in this field are often guided by the symmetry principle. 
For example, it has been realized that in crystalline insulators the lattice translation invariance (a unitary $\Z^d$-symmetry) 
and the time-reversal symmetry (an antiunitary $\Z_2$-symmetry) cooperate to yield 
exponential localization of the electronic states.  
More precisely, in a sequence of papers starting with a seminal contribution by the Nobel Laureate W. Kohn 
\cite{Kohn59,Cl1,Cl2,NeNe82,Ne83,HeSj89,BrPaCaMoMa2007,Panati2007,FiMoPa_2,CaLePaSt2016,CoHeNe2015,CoMoMos2018}, it has been rigorously proved that 
- if the Hamiltonian is gapped and enjoys both the mentioned symmetries - one can construct Wannier functions associated 
to the Bloch bands  below the Fermi energy which are exponentially localized.  
We refer to [MoPa] for a short review of this subject. 

Whenever a magnetic field is included in the model, time-reversal symmetry (TRS) is broken, since the 
corresponding Schr\"{o}dinger operator does not commute with complex conjugation, which represents TRS in this simple one-particle context. 
Hence, a natural question arises: how the breaking of TRS reflects in the localization properties of the 
electronic states?  

One may address this question at different levels.  On the one hand, an abstract, model-independent answer has 
been recently provided \cite{MoPaPiTe1,MoPaPiTe2}, and advertised under the name of  
\emph{Localization-Topology Correspondence}.  In a nutshell, this recent result says that whenever the space of states below the 
Fermi energy has a non-trivial topology, there is no choice of Wannier functions yielding a finite expectation value 
for the squared position operator. In view of the Transport-Topology Correspondence \cite{TKNN}, this topological non-triviality reflects in a non-zero transverse conductance at zero temperature, hence in a measurable property of the system under investigation.  The Localization-Topology Correspondence relates the latter property to the localization of the electronic state.

On the other hand, here we investigate the same question in a more explicit way, 
by considering the simplest possible model in the suitable symmetry class: the Landau Hamiltonian.  
We refer here to the Altland-Zirnbauer symmetry classes, which are commonly used to classify topological 
insulators and superconductors  \cite{AtZi1997,Kitaev2009,HasanKane}.
The Landau, the Hofstadter and the Haldane operators are, presumably, the simplest models in the symmetry class A, 
which describes Quantum Hall systems and Chern insulators. While the latter two examples refer to discrete models, 
the former is a differential operator, hence tractable with the usual tools of Real and Functional Analysis. 
Notice that the application of the {Localization-Topology Correspondence} to the Haldane model has been recently discussed in \cite{MoMaMoPa}.

The answer we provide is actually not completely new, as the same problem has been investigate by 
Zak and coworkers \cite{Zak97,RaZuEf1997,Zak98}. However, our approach might be more appealing to mathematically oriented readers, since our proof clearly highlights the central role of the Balian-Low theorem as the deep mathematical structure behind the result. In this paper, following the idea in \cite{Zak97,Zak98}, we provide an alternative and self-consistent proof of the Zak's results using an abstract version of the Balian-Low theorem.

After reviewing the Landau model, with a particular emphasis on its symmetries and its relation 
with the Segal-Bargmann holomorphic representation of Quantum Mechanics (Section \ref{Sec:Landau}),
we provide a new,  representation-independent proof of the Balian-Low theorem, which slightly generalizes 
Battle's proof (Section \ref{Sec:Balian-Low}). As a consequence of this abstract Balian-Low Theorem (Theorem \ref{BalianLowTheorem}) we obtain a straightforward and 
mathematically transparent  proof of Zak's results.   
In short, we prove that it is not possible to construct a system of orthonormal eigenfunctions
corresponding to a given Landau level which are both intertwined by magnetic translations 
(the technical word is \emph{Gabor frame}) and well-localized in space.   

The non-existence of such a well-localized basis for a single Landau level is in complete accordance with the paradigm of the Localization-Topology Correspondence. Indeed, each Landau level is topologically non trivial and its Hall conductivity is non-zero \cite{AvSeSi1994}.

Although the same results on the Landau Hamiltonian have been already obtained by Zak two decades ago, we are confident that our new argument may be useful to clarify the essential mathematical structure underlying the subtle relation between broken time-reversal symmetry and delocalization 
in magnetic periodic quantum systems.   

\section{The Landau model and its symmetries: a review}
\label{Sec:Landau}
In this Section we shortly review  the main properties of the Landau Hamiltonian, with emphasis on its symmetries and on the correspondence between the classical and the quantum theory. We are interested in describing the quantum mechanical behaviour of a point charge in $\R^2$ under the influence of a magnetic field perpendicular to the plane. 


At the classical level one considers the phase space $\R^{2}\times \R^2$ where each point is labelled as $(q,p)$ with $q$ and $p$ vectors in $\R^2$. The Hamiltonian governing the dynamics of the system is 
$$
H(q,p)=\frac{1}{2 m}\left( \vphantom{\frac{1}{2}}p-\frac{ \text{\tiny Q}}{c} \, {A}\sub{ph}(q)\right)^2 \,,
$$
where $m$ and $\text{\tiny Q}$ denote, respectively, the mass and the charge of the particle, and $c$ is the speed of light. For a uniform magnetic field $B$, the corresponding vector potential (in the usual physical units) is 
$$
{A}\sub{ph}(q)=\frac{1}{2} \left( {B} \cdot {\mathbf{e}_3} \right) \left( - q_2, q_1 \right)  
$$
where $\left({\mathbf{e}_1},{\mathbf{e}_2},{\mathbf{e}_3}\right)$  is a positively-oriented basis with ${\mathbf{e}_3} \mathbin{\!/\mkern-5mu/\!} {{{B}}}$ (clearly the embedding $\R^2 \subset \R^3$ is understood).   For the sake of a simpler notation, we consider a particle of unit mass, so that the Hamiltonian reads
\begin{equation}
\label{CHamiltonian}
H(q,p) =\frac{1}{2}\left( \vphantom{\frac{1}{2}}p-A(q)\right)^2 
\end{equation}
where $A(q)=\frac{b}{2}  \left(-q_2,q_1\right)$, with $b=\frac{ \text{\tiny Q}}{c} \left( {{B}} \cdot {\mathbf{e}_3} \right)$. Up to an appropriate choice of the orthonormal frame, one can always assume that $b>0$, as we do hereafter. Notice, however, that the dynamics depends on the sign of the charge, since for positive charges one has $\left( {B} \cdot {\mathbf{e}_3} \right)>0$, while 
$\left({B} \cdot {\mathbf{e}_3} \right) <0$ for negative charges.

The dynamics is generated by Hamilton's equations 
$$
\left\{ \begin{aligned}
&\dot{q_1}=\frac{\partial H }{\partial p_1} (q,p)= p_1 + \frac{b}{2} q_2 \, ,\\
& \dot{q_2}=\frac{\partial H }{\partial p_2} (q,p)= p_2 - \frac{b}{2} q_1 \, ,\\
&\dot{p_1}=-\frac{\partial H}{\partial q_1} (q,p)= \frac{b}{2} (p_2 - \frac{b}{2} q_1 ) \, ,\\
& \dot{p_2}=-\frac{\partial H}{\partial q_2} (q,p)= -\frac{b}{2} (p_1 + \frac{b}{2} q_2 ) \, .
\end{aligned} \right.
$$
By substitution, it follows that 
$$
\left\{
\begin{aligned}
&\ddot{q_1} = b  \,\dot{q_2} \, ,\\
&\ddot{q_2} = -b \, \dot{q_1} \, .
\end{aligned}
\right.
$$
The solutions are described by the well-known cyclotron dynamics
\begin{equation}
\label{CiclotronOrbit}
\left\{
\begin{aligned}
&q_1(t)=\bar{q}_1-\frac{1}{b} \, \dot{q_2}(t) \, ,\\
&q_2(t)=\bar{q}_2+\frac{1}{b} \, \dot{q_1}(t) \, 
\end{aligned}
\right.
\end{equation}
where we denoted by $\bar{q}$ the centre of the cyclotron orbit, which depends on the initial data $\left( q(0), \dot{q}(0) \right)$.

\subsection{Symmetries of the classical theory}
\label{Sec:classical}
From the equations \eqref{CiclotronOrbit} it is easy to recognize that the centre of the cyclotron orbit is a constant of motion. This can be understood as follows: the solutions to the equations of motion are localized in space, but the system has no preferred centre of localization. The area in which the solution is localized is labelled by the coordinates of the centre $\bar{q}=\left(\bar{q}_1,\bar{q}_2\right)$.

A posteriori we notice that the following canonical transformation simplifies the problem in a substantial way:
\begin{equation}
\label{Classicalcoord}
\left\{
\begin{aligned}
&\widetilde{q}_1=\bar{q}_1=\frac{1}{b} (p_2+ \frac{b}{2}q_1) \, , \qquad &&\widetilde{p}_1=\bar{q}_2=\frac{1}{b} (-p_1+ \frac{b}{2}q_2) \, ,\\
&\widetilde{q}_2= \frac{1}{b}(p_1 + \frac{b}{2} q_2) \, , \qquad \, &&\widetilde{p}_2=p_2 - \frac{b}{2} q_1  \, .
\end{aligned}
\right.
\end{equation}
Indeed, after the above transformation the Hamiltonian depends only on the second set of coordinates $$H(\widetilde{q},\widetilde{p})=\frac{1}{2}({\widetilde{p}_2}^{\,\, 2} + b^2 \, {\widetilde{q}_2}^{\,\, 2})$$ and it is simply the Hamiltonian of a harmonic oscillator, but defined on the phase space $\R^2 \times \R^2 \ni \left(\widetilde{q}_1, \widetilde{q}_2, \widetilde{p}_1,\widetilde{p}_2\right)$. In particular, the following Poisson brackets vanish
$$
\left\{H,\widetilde{q}_1\right\}=0=\left\{H,\widetilde{p}_1\right\} \, ,
$$
as a consequence of the special symmetry of the problem. Let us see how this information reflects at the quantum mechanical level.

\subsection{Symmetries of the quantum theory}
\label{Sec:quantum}

Consider the Hilbert space $L^2(\R^3)$, let ${\bf{P}}=-\iu \nabla$ be the momentum operator and let $X_i$, $i\in \left\{1,2,3\right\}$, be the three components of the position operator, \ie $(X_i \psi)(\x)=x_i \psi(\x)$ for all $\psi \in \Dom(X_i)$. The Schr\"odinger operator that describes a point charge of mass $m=1$ moving in $\R^3$ under the influence of a constant magnetic field perpendicular to the $xy$ plane is given by \footnote{ We use Hartree units (so that  $e^2$, $\hbar$, $m_e$ are dimensionless and equal to $1$) and $b$ is defined as in \eqref{CHamiltonian}.} 
$$
\iu \partial_t \psi = \frac{1}{2} \left\{ \left( P_1 + \frac{b}{2}X_2 \right)^2 + \left( P_2 - \frac{b}{2}X_1 \right)^2 + P_3^2  \right\} \psi =: H \psi\, .
$$
Exploiting the isomorphism $L^2(\R^3) \cong L^2(\R^2)\otimes L^2(\R)$ it is sufficient to study the dynamics induced by the following operator, densely defined in $L^2(\R^2)$,
\begin{equation*}
H_L := \frac{1}{2}  \left( P_1 + \frac{b}{2}X_2 \right)^2 + \frac{1}{2}  \left( P_2 - \frac{b}{2}X_1 \right)^2 \, ,
\end{equation*}
which is called Landau Hamiltonian. Note that $H_L$ is essentially selfadjoint on the dense domain $C^{\infty}_0(\R^2)$. Motivated by the classical case, see \eqref{Classicalcoord}, we define the following operators
\begin{equation}
\label{QDef}
\begin{aligned}
&\KM_1:= \frac{1}{b}(P_1 + \frac{b}{2} X_2) \, , \qquad \, \KM_2:=P_2 - \frac{b}{2} X_1  \, ;\\
&\OC_1:=\frac{1}{b} (P_2+ \frac{b}{2}X_1) \, , \quad \OC_2:=\frac{1}{b} (-P_1+ \frac{b}{2}X_2) \, .
\end{aligned}
\end{equation}
The operators $\left(\OC_1,\OC_2\right)$ are the quantum analogous of the coordinates of the cyclotron orbit $\left(\bar{q}_1,\bar{q}_2\right)\equiv \left(\widetilde{q}_1,\widetilde{p}_1\right) $ in the classical theory, while $\left(\KM_1,\KM_2\right)$ are the \virg{dynamical coordinates} $\left(\widetilde{q}_2,\widetilde{p}_2\right)$ appearing in \eqref{Classicalcoord}.
As before, we can write the Hamiltonian as $H_L=\frac{1}{2}\left({\KM_2}^{2} + b^2 {\KM_1}^{2} \right)$.  All the operators defined in \eqref{QDef} are essentially selfadjoint on $C^{\infty}_0(\R^2)$, see \cite[Proposition 9.40]{Hall13}, and by explicit computations we observe that they satisfy the following commutation relations
\begin{equation}
\label{QCR}
\begin{aligned}
&\left[\KM_1,\KM_2\right]= \iu \Id \, ,  \qquad &&\left[\OC_1,\OC_2\right]= -\frac{\iu}{b} \Id \, , \\
&\left[\KM_j, \OC_l\right]= 0 \, , \qquad &&\;\forall \, j,l \in \left\{1,2\right\} \, .
\end{aligned}
\end{equation}
In view of Stone's Theorem, we have the following proposition. 

\begin{prop} \cite[Proposition 13.5]{Hall13} 
	\label{StoneT}
	The four operators defined in equations \eqref{QDef} are the four generators of the following one parameter unitary groups
	\begin{align}
	\label{StoneK}
	&\{e^{\iu t \KM_1}= e^{\iu t\frac{1}{b}P_1} e^{\iu t\frac{1}{2} X_2} \}_{t \in \R} \, , \qquad \, \{e^{\iu t \KM_2}= e^{\iu t P_2} e^{-\iu t\frac{b}{2} X_1} \}_{t \in \R}  \, ;\\
	\label{StoneG}
	&\{e^{\iu t \OC_1}= e^{\iu t\frac{1}{b}P_2} e^{\iu t\frac{1}{2} X_1} \}_{t \in \R} \, , \quad \{e^{\iu t \OC_2}= e^{-\iu t\frac{1}{b}P_1} e^{\iu t\frac{1}{2} X_2} \}_{t \in \R} \, .
	\end{align}
	
\end{prop} 

Let us focus our attention on the action of the one parameter unitary group associated to $\OC_1$. For any vector $\psi \in L^2(\R^2)$, since $[X_1, P_2] =0$, one has that 
\begin{equation}
\label{ActionSymMT}
(e^{\iu \alpha \OC_1} \psi)(\x)=e^{\iu \alpha \frac{x_1}{2}} \psi\left(x_1,x_2+\frac{\alpha}{b}\right) \, .
\end{equation}
We see that the action of $e^{\iu \alpha \OC_1}$ amounts to a translation of the wave function and a gauge transformation generated by the function $\chi(\x):=-\alpha \frac{x_1}{2}$. Hence the magnetic vector potential in the new gauge is given by
$$
A'(\x)=A(\x)+ \nabla \chi(\x) = \frac{b}{2}\left(-\left(x_2+\frac{\alpha}{b}\right),x_1\right)=A\left(x_1,x_2+\frac{\alpha}{b}\right)\, 
$$
that is exactly the same translation performed on the wave function. In other words, what the unitaries $e^{\iu \alpha \OC_i}$ are doing is shifting the eigenfunctions in the plane $\R^2$, which means that they are changing the centre of localization of the wave function and simultaneously acting with a \emph{gauge transformation} that shifts the zero of the magnetic vector potential by the same amount.  What is the physical interpretation of this fact and how is it related to the classical cyclotron orbit centre?  

While the Hamiltonian $H_L$ does not commute with the ordinary translations $T_{\boldsymbol{\alpha}}$ (defined by $\left(T_{\boldsymbol{\alpha}}\psi\right) (\x)=\psi(\x-\boldsymbol{\alpha})$ ) it does commute with the \emph{ magnetic translations } $e^{\iu \alpha_2 \OC_1}$ and $e^{\iu \alpha_1 \OC_2}$, as stated in the following proposition, which is a standard results and whose proof is detailed for the reader's convenience.

\begin{prop}
	The one-parameter unitary groups defined in \eqref{StoneG} are two symmetry groups for the Hamiltonian $H_L$, that is
	$$
	e^{\iu \alpha \OC_j} H_L e^{-\iu \alpha \OC_j} = H_L \, , \quad \forall \alpha \in \R \, , \quad j\in \left\{1,2\right\}\,.
	$$
\end{prop}
\begin{proof}
	Consider a vector $\varphi \in C^{\infty}_0(\R^2)$. It is easy to see that $H_L \varphi \in C^{\infty}_0(\R^2)$ and $\OC_j \varphi \in C^{\infty}_0(\R^2)$ for $j\in \left\{1,2\right\}$. Then we are allowed to define the commutator as an operator acting on the dense invariant subset $C^{\infty}_0$ and we get
	$$
	H_L \OC_j \varphi - \OC_j H_L \varphi = \left[H_L,\OC_j\right] \varphi =0\, .
	$$
	Hence the commutator is well defined on $L^2(\R^2)$ and is equal to the null operator. Consider now the one-parameter unitary group given by the Stone Theorem, that is $e^{\iu \alpha \OC_j}$, $\alpha \in \R$. From \cite[Lemma 10.17]{Hall13} we have that for every $\varphi \in C^{\infty}_0$, the function 
	$$
	S_{\varphi}(\cdot) : \alpha \mapsto e^{\iu \alpha \OC_j} H_L e^{-\iu \alpha \OC_j}\varphi \; , 	S_{\varphi}(\cdot) : \R \to L^2(\R^2)
	$$
	is differentiable and 
	$$
	\partial_{\alpha}S_{\varphi}(\alpha)=\iu e^{\iu \alpha \OC_j} \left[\OC_j,H_L\right] e^{-\iu \alpha \OC_j}\varphi = 0 \, 
	$$
	because the space $C^{\infty}_0$ is invariant by $e^{-\iu \alpha \OC_j}$ \footnote{See formula \eqref{ActionSymMT}. The action of the unitary group is simply a translation and a multiplication by a smooth phase.} and then we can use the argument above. Therefore we have that
	$$
	e^{\iu \alpha \OC_j} H_L e^{-\iu \alpha \OC_j} = H_L \, ,
	$$
	and the group properties are guaranteed by Proposition \ref{StoneT}.
\end{proof}

Using the standard approach, whose essential idea traces back to Landau, one can prove the following two propositions. 
\begin{prop}[Landau]
	\label{SpectrumHL}
	The spectrum of the Landau Hamiltonian $H_L$ is discrete and given by
	$$
	E_n=|b|\left(n+\frac{1}{2}\right) \, , \qquad n \in \N \, ,
	$$ 
	where  $E_n$ is called the $n^{th}$ Landau level\footnote{With a little abuse of terminology, we use the term \virg{$n^{th}$ Landau level} or \virg{lowest Landau level} also to refer to the corresponding eigenspaces.}. Moreover the function 
	\begin{equation}
	\label{SolutionAZero}
	\phi_0(\x)= \left(\frac{|b|}{2 \pi} \right)^{\frac{1}{2}} e^{-\frac{|b|}{4}|\x|^2}
	\end{equation}
	is an eigenfunction corresponding to the lowest Landau level $E_0$.
\end{prop}
For comparison purposes, we briefly sketch the main steps of the proof of the previous proposition.
\begin{proof}
	Recall that, with our conventions, one has $b>0$. Consider the ladder operators
	\begin{align*}
	&A=\frac{1}{\sqrt{2b}} \left(b \KM_1 + \iu \KM_2 \right) \\
	&A^* = \frac{1}{\sqrt{2b}} \left(b \KM_1 - \iu \KM_2 \right) \, .
	\end{align*}
	One checks that $\left[A,A^*\right]=\Id$ and 
	$$
	H_L= b \left(A^* A + \frac{1}{2} \right) \, .
	$$
	This is exactly what happens in the case of the harmonic oscillator, where the positive operator $A^* A$, namely the number operator, has discrete spectrum, $\sigma(A^* A)=\N$, and the action of  $A^*$, the raising operator (resp.$\ A$, the lowering operator) is to \virg{move up} (resp.$\ $\virg{move down}) in the spectrum. Namely, if $\varphi$ is such that $A^* A \varphi = n\varphi$ then $(A^* A) A^* \varphi=(n+1)A^* \varphi$ and $(A^* A) A \varphi=(n-1)A \varphi$. Therefore it suffices to show that the kernel of the number operator contains non-null vectors, namely that there exists a function $\varphi_0$ such that
	\begin{equation}
	\label{AZero}
	A \varphi_0 = 0 \, .
	\end{equation}
	One can easily check that the function defined in \eqref{SolutionAZero} satisfies \eqref{AZero}.
	
	The previous argument shows that $\N \subset \sigma_{\rm pp}(A^*A)$. The inverse inclusion is obtained by means of a simple argument. Suppose that there exists a vector $\varphi_{\lambda}\neq 0$ such that $A^*A \varphi_{\lambda}= \lambda \varphi_{\lambda}$, $\lambda>0$. Then for $m\in \N$, $m>\lambda$ we have only two cases: either $A^m \varphi_{\lambda} \neq 0$ or $A^m \varphi_{\lambda} = 0$. Since 
	$$
	A^* A \left( A^m \varphi_{\lambda} \right) = (\lambda-m) \left( A^m \varphi_{\lambda} \right) \, ,
	$$
	the first case implies the existence of a negative eigenvalue, which is impossible in view of the non-negativity of $A^* A$. Regarding the second case, let $\varphi_{\lambda-m+n}:= A^{m-n} \varphi_{\lambda}$, where $n>0 $ is the smallest integer such that $A^{m-n} \varphi_{\lambda} \neq 0$. Notice that $n\leq m$ by hypothesis. Then $A^* A\varphi_{\lambda-m+n}= (\lambda -m +n ) \varphi_{\lambda-m+n}$ with $\varphi_{\lambda-m+n}\neq 0$. By acting with the linear operator $A^*$ on $A \varphi_{\lambda-m+n} = 0$ one gets
	$$
	0 = A^* A \varphi_{\lambda-m+n} = (\lambda -m +n ) \varphi_{\lambda-m+n} \, ,
	$$ 
	which forces $\lambda$ to be equal to $m-n \in \N$. This proves that $ \sigma_{\rm pp}(A^*A)=\N$.
	
	A completeness argument, which we omit for the sake of brevity, shows that $\sigma(A^*A)= \sigma_{\rm pp}(A^*A)=\N$, which concludes the proof.
\end{proof}

Notice that, even if the Landau Hamiltonian has the same spectrum as the harmonic oscillator Hamiltonian, the spectral type is different.
\begin{prop}
	\label{InfiniteDeg}
	Each eigenvalue $E_n$ of $H_L$ is infinitely degenerate. Hence the spectrum of $H_L$ is purely essential spectrum.
\end{prop}
\begin{proof}
	Recall that, with our conventions, one has $b>0$. The operators $\OC_1$ and $\OC_2$ satisfy the commutation relation \eqref{QCR} and, in the same spirit of the proof of Proposition \ref{SpectrumHL}, we can use them to construct another set of ladder operators, namely
	\begin{align*}
	&B=\sqrt{\frac{b}{2}} \left(\OC_1 - \iu \OC_2 \right) \\
	&B^* = \sqrt{\frac{b}{2}} \left( \OC_1 + \iu \OC_2 \right) .
	\end{align*}
	One can easily check that $\left[B,B^* \right]= \Id$. Similarly to what happens for the ladder operators $A$ and $A^*$, we have that the number operator $B^* B$ has discrete spectrum and the role of the raising and lowering operators is played now by $B^*$ and $B$ respectively. 
	
	By direct computation one can check that the eigenfunction $\varphi_0$ defined in \eqref{AZero} satisfies also
	$$
	B \phi_0 = 0.
	$$
	This means that $\phi_0$ is in the kernel of the number operator $B^* B$. Since the operators $\OC_i$ commute with the operators $\KM_i$, we have that
	$$
	A^* A (B^*)^n \phi_0 = (B^*)^n A^* A  \phi_0 = 0 \, .
	$$
	Therefore all the infinite eigenvectors of the number operator $B^* B$ are eigenvectors of $H_L$ corresponding to the lowest Landau level. Since $B^*B$ is a selfadjoint operator, eigenvectors corresponding to different eigenvalues are orthogonal. This proves that the eigenvalue $E_0$ is infinitely degenerate. A similar argument shows that each Landau level $E_n$, $n \in \N$, is infinitely degenerate.
\end{proof}

Note that the expectation value of the number operator $B^*B$ on a given state \virg{measures} the distance from the origin of the state, since $\frac{2}{b}\left(B^*B+\frac{1}{2}\right)=\OC_1^2+\OC_2^2$ is the analogous of the square of 
the distance of the centre of the cyclotron orbit from the origin, in the classical theory.
Classically, there exist infinitely many orbits with the same cyclotron radius, which differ among each other by the position of the centre of the orbit. This infinite multiplicity of the solutions of the classical dynamics reflects into the infinite degeneracy of the Landau levels.

The eigenvector $\phi_0$ defined in Proposition \ref{SpectrumHL} is well-localized around the origin. Nevertheless, the action of the symmetry groups generated by $\OC_1$ and $\OC_2$ implies that, for a given energy, there is no preferred centre of localization. In other words, the energy does not depend on where the eigenfunction is localized. Indeed, given an eigenfunction $\psi_n$ of $H_L$ associated with the eigenvalue $E_n$, we have 
$$
H_L  e^{\iu \alpha \OC_i} \psi_n = e^{\iu \alpha \OC_i} H_L \psi_n =   E_n  e^{\iu \alpha \OC_i} \psi_n \, .
$$
Hence $e^{\iu \alpha \OC_i} \psi_n$ is an eigenfunction of $H_L$ associated with the eigenvalue $E_n$.

The general translation of the orbit centre is described in the next definition. 

\begin{dfn}
	For every vector $\boldsymbol{\alpha}=(\alpha_1,\alpha_2) \in \R^2$ we define the \emph{magnetic translation operator} associated to $\boldsymbol{\alpha}$ to be the following unitary operator 
	$$
	\tau^{(b)}_{\boldsymbol{\alpha}}:=e^{\iu b(\alpha_1 \OC_2  -  \alpha_2  \OC_1 )} \, . 
	$$
	
\end{dfn}

\subsection{Relation with the Segal-Bargmann representation}
\label{Sec:Segal}

Two strongly continuous families of unitary operators $\left\{U(t)\right\}_{t \in \R}$ and $\left\{V(s)\right\}_{s \in \R}$ on a given Hilbert space $\Hi$ satisfy the Weyl relations if
\begin{equation}
\label{WeylUV}
U(t)V(s)=e^{-\iu t s} V(s)U(t) \, , \qquad \forall \, t,s \in \R\, .
\end{equation}
Irreducible representations of the Weyl relations play an important role in the mathematical formulation of quantum mechanics. The basic example of irreducible representation of the Weyl relations is given by the unitary group associated with the canonical position operator $X$ and the canonical momentum operator $P$ acting in $L^2(\R)$. In this section we show that the Landau model provides an infinite number of irreducible representations of the Weyl relations, namely one for each Landau level. While we explicitly discuss the case of the Lowest Landau Level (LLL), the same argument holds true with minor modification also for any other Landau level.
\medskip

Consider the ladder operator $B$ defined in Proposition \ref{InfiniteDeg}. Explicitly, $B$ acts as
$$
\left(B\psi\right)(x,y)=\sqrt{\frac{1}{2b}} \left( \frac{\partial \psi }{\partial x_1}(x,y) -i\frac{\partial \psi}{\partial x_2} (x,y)  + \frac{b}{2} (x_1 -\iu x_2) \psi (x,y)\right) , \forall \psi \in C^\infty_0(\R^2) \, .
$$
Identifying $\R^2$ with the complex plane, namely $x_1 + \iu x_2=:z \in \C $, and setting
$$
\partial := \frac{1}{2} \frac{\partial}{\partial z}  =\frac{1}{2} \left( \frac{\partial }{\partial x_1} -i\frac{\partial}{\partial x_2} \right) \, , \quad  \bar{\partial} := \frac{1}{2} \frac{\partial}{\partial \bar{z}}= \frac{1}{2} \left( \frac{\partial }{\partial x_1} +i\frac{\partial}{\partial x_2}  \right)\, ,
$$ 
we have that
$$
B \psi(z) = \sqrt{\frac{1}{2b}} \left( 2\partial + \frac{b}{2} \bar{z} \right) \psi(z) \, 
$$
and similarly
$$
B^{*} \psi(z) = \sqrt{\frac{1}{2b}} \left( -  2\bar{\partial} + \frac{b}{2} z \right) \psi(z) \, .
$$

If we substitute $\psi$ with the eigenvector $\phi_0$ defined in \eqref{SolutionAZero}, we get that
$$
B \phi_0(z) = 0 \, , \qquad  B^{*} \phi_0(z) = \sqrt{\frac{b}{2}} \, z \phi_0(z) \, ,
$$
This means that the action of the raising operator $B^*$ on $\phi_0$ amounts to multiplication by $z$. Therefore, since the LLL is a closed subspace, we get that a generic function in the LLL is of the form 
$$
\psi(z)= f(z) \phi_0(z)
$$
where $f(z)$ is analytic and such that 
$$
\int_{\C} dz |f(z)|^2 |\phi_0(z)|^2 < \infty \, .
$$

From the commutation relation $\left[B,B^*\right]=\Id$ one can deduce that the action of $B$ and $B^*$ can be described only in terms of the analytic function $f$, that is
$$
\left(B f \phi_0\right)(z) = \sqrt{\frac{2}{b}} \left( \partial f(z) \right) \phi_0(z) \, , \qquad  	\left(B^{*} f \phi_0\right)(z) = \sqrt{\frac{b}{2}} \,  z f(z) \phi_0(z) \, .
$$

To be more precise, one considers the Gaussian measure $d\mu := N \expo^{-\frac{|b|}{4}|z|^2} dz$, with $N$ positive constant, and defines the weighted $L^2$-space
\begin{equation*}
L^2(\C,d\mu) := \left\{g: \C \to \C : \int_{\C}|g(z)|^2 d\mu < \infty\right\}
\end{equation*} 
endowed with the scalar product
\begin{equation*}
\langle f,g \rangle_{SB} := \int_{\C} \overline{f(z)} g(z) d\mu(z) \, .
\end{equation*}

\begin{dfn}[Segal \cite{Segal15}, Bargmann \cite{Ba1961}]
	Let $\Hol(\C)$ be the space of entire functions on $\C$. The Segal-Bargmann space $SB(\C)$ is defined as 
	\begin{equation*}
	SB(\C):=\left\{ g\in \Hol(\C) : \int_{\C}|g(z)|^2 d\mu(z) < \infty \right\} \;=\; L^2(\C,d\mu(z)) \cap \Hol(\C) \;.
	\end{equation*}
\end{dfn}
It is straightforward to identify the LLL and the Segal-Bargmann space via the unitary operator $U:\Pi_0 L^2(\R^2) \to SB(\C) $ defined by
$$
(U \psi)(z) = f(z) \, 
$$
where $\psi(x,y)=f(x,y) \phi_0 (x,y)$, and $\Pi_0$ denotes the projection onto the LLL. Therefore we obtain
$$
 U \Pi_0 B \Pi_0 U^* = \sqrt{\frac{2}{b}} \,  \partial  \, , \qquad U \Pi_0 B^* \Pi_0 U^* = \sqrt{\frac{b}{2}} \,  z \, .
$$
Notice that the operators $\OC_i$ are related to the operator $z$ and $\partial$ by the following relations
\begin{equation}
\label{UG}
\begin{aligned}
&U\Pi_0(B+B^*)\Pi_0U^* = U \Pi_0 (\sqrt{2b} \OC_1) \Pi_0 U^* =  \sqrt{\frac{1}{2b}} \left( 2 \partial + b z\right) \, , \\
&U\Pi_0(B^*-B)\Pi_0U^* = U\Pi_0 ( \iu \sqrt{2b} \OC_2) \Pi_0U^* =  \sqrt{\frac{1}{2b}} \left( b z- 2\partial\right) \,.
\end{aligned}
\end{equation}

The relations \eqref{UG} together with \cite[Theorem 14.16]{Hall13} allow to prove that the operators $e^{\iu t\OC_1}$, $t \in \R$, and $e^{\iu s\OC_2}$, $s \in \R$, form an irreducible representation of the Weyl relations. The irreducible representation space is provided by (the eigenspace corresponding to) the lowest Landau level.

We emphasize that the Segal-Bargmann representation, by means of its complex plane formalism, provides a simple and straightforward characterization of the lowest Landau level in terms of entire functions and, as a by-product, also the action of the operators $B$ and $B^*$ is extremely simplified. This formalism turned out to be useful in the study of Fractional Quantum Hall effect, in particular in studying the effect of external potentials acting on the Landau levels, see \cite{GiJa84, MoscolariPanati}.

\section{An abstract Balian-Low theorem}
\label{Sec:Balian-Low}
It emerges from the previous discussion that the pair of operators $\left(\OC_1,\OC_2\right)$ defined in the previous Section has some structural analogies with the canonical pair $\left(X,P\right)$, given by position and momentum operators in the Schr\"odinger representation. The essential structure is captured by the following abstract definitions of \emph{Gabor triple} and \emph{Generalized Gabor frame} (GGF).

Given any Hilbert space $\Hi$ endowed with the sesquilinear form $\langle \cdot , \cdot \rangle : \Hi \times \Hi \to \C$, consider two selfadjoint operators $\X$ and $\K$ (whose domains are denoted by $\Dom(\X)$ and $\Dom(\K)$, respectively), a dense subspace $\Dd \subset \Hi$ and a two dimensional lattice $\Gamma$ generated by the vectors  $\aU, \aD \in \R^2$ via
\begin{equation*}
\Gamma:= \left\{ \gamma=n \aU + m \aD  \in \R^2 : n,m \in \Z \right\} \subset \R^2 \, .
\end{equation*}
Our main Assumption is the following:
\begin{assumption}
	\label{Assumption}
	Assume that $\X,\K$ and $\Dd$ satisfy the following properties :
	\begin{enumerate}[label=(\roman*), ref=(\roman*)]
		\item  \label{Domain} $\Dd$ is a \emph{common core} for $\X$ and $\K$. This means that for every vector $\psi \in \Dom(\X)~\cap~\Dom(\K) $ there exists a sequence $\{\xi_i\} \subset \Dd$ such that \footnote{In the definition of \emph{common core} it is essential that the sequence $\xi_i \to \psi$, $i \to \infty$ provides convergence of both  $\left\{\X \xi_i\right\}_{i \in \N}$ and $\left\{\K \xi_i\right\}_{i \in \N}$. This is, in general, stronger than asking that $\Dd$ is a core for both $\X$ and $\K$. }
		\begin{equation}
		\label{ApproxD}
		\begin{aligned}
		&\xi_i \to \psi \\
		&\X \xi_i \to \X \psi \qquad \qquad \text{ as }  i \to \infty \\
		&\K \xi_i \to \K \psi \, ,
		\end{aligned}
		\end{equation}
		where the convergence is understood in the norm of $\Hi$.
		\item The operators $\X$ and $\K$ satisfy the Weyl commutation relations, that is, for all $t,s \in \R$
		\begin{equation}
		\label{Weylrelation}
		e^{\iu t \X} e^{\iu s \K} = e^{-\iu ts } e^{\iu s \K} e^{\iu t \X}  \, ,
		\end{equation}
		compare with equation \eqref{WeylUV}.
	\end{enumerate}
\end{assumption}

\begin{dfn}[Gabor triple]
	\label{Gabor3}
	A Gabor triple $\Gabor:=\left(\X,\K,\Dd\right)$ consists of two self-adjoint operators $\X$, $\K$ and a  dense subspace $\Dd \subset \Hi$ such that Assumption \ref{Assumption} holds true.
\end{dfn}

From the Weyl commutation relations we have the following well-known results.
\begin{prop}
	\label{PropT}
	From Assumption \ref{Assumption}.(ii) it follows that 
	\begin{enumerate}[label=(\roman*),ref=(\roman*)]
		\item The map $T: \C
		\to \mathcal{U}(\Hi)$ defined by
		\begin{equation*}
		T(z):=e^{\iu \frac{\re(z) \im(z)}{2}}e^{\iu \re(z) \X}e^{\iu \im(z) \K} \, 
		\end{equation*}
		for every $z\in \C$, defines a projective unitary representation of the additive group $\C$, that is
		\begin{equation*}
		T(z)T(z')=e^{-\iu \frac{\im  z \wedge z' }{2}}\, \,T(z+z') \, ,
		\end{equation*}
		where $ z \wedge z' :=\iu \left(\re{z}\im{z'}-\re{z'}\im{z}\right)$.
		
		\item \label{ProjRep} The restriction of the map $T$ to the lattice $\Gamma$ (after the identification of $\R^2$ with $\C$) is a projective unitary representation of $\Z^2$.
	\end{enumerate}
\end{prop}

\begin{dfn}[Generalized Gabor Frame]
	Consider a Gabor triple $\Gabor$, a lattice $\Gamma$ generated by the vectors $\aU, \aD \in \R^2$, a closed subspace $\mathcal{V} \subset \Hi$ and an element $\varphi_0 \in \mathcal{V}$. If the set 
	\begin{equation}
	\{\varphi_{m,n} \}_{m,n \in \Z}= \{ T(\aU)^{m} \, T(\aD)^{n} \varphi_0 \}
	\end{equation}
	is contained in $\mathcal{V}$, we call it a Generalized Gabor Frame (GGF) for $\mathcal{V}$ generated by $\varphi_0$ and associated to the Gabor triple $\Gabor$ and the lattice $\Gamma$.
\end{dfn}
Hereafter we make use of the short-hand notation $T_{m,n}:= T(\aU)^{m} \, T(\aD)^{n} $.
\begin{prop}
	\label{HeisenbergCR}
	Let $\Gabor=\left(\X,\K,\Dd\right)$ be a Gabor triple. Then $\X$ and $\K$ satisfy the weak canonical commutation relations, that is $\forall\, \psi,\varphi \in \Dom(\X)\cap \Dom(\K) $
	\begin{equation}
	\label{WeakCCR}
	\langle \X \psi, \K \varphi \rangle - \langle  \K \psi, \X \varphi \rangle= \iu 	\langle \psi, \varphi \rangle.
	\end{equation}
\end{prop}
\begin{proof}
	Consider two vectors $\psi,\varphi \in \Dom(\X)\cap \Dom(\K) $. By Assumption \ref{Assumption} \ref{Domain} we know that there exist two sequences $\left\{\xi_i\right\}_{i\in \N} \subset \Dd $ and $\left\{\zeta_i\right\}_{i\in \N} \subset \Dd$ satisfying \eqref{ApproxD} for $\psi$ and $\varphi$, respectively. From \eqref{Weylrelation} we have
	\begin{equation}
	\label{Aux1}
	\langle e^{-\iu t \X} \xi_i ,  e^{\iu s \K} \zeta_i \rangle = e^{-\iu ts } \langle e^{-\iu s \K}  \xi_i,  e^{\iu t \X} \zeta_i \rangle \, .
	\end{equation}
	Define now $F_i(t,s):=\langle e^{-\iu t \X} \xi_i ,  e^{\iu s \K} \zeta_i \rangle $ and $\widetilde{F}_i(t,s):=e^{-\iu ts } \langle e^{-\iu s \K}  \xi_i,  e^{\iu t \X} \zeta_i \rangle$. By Stone's Theorem and the hypothesis on $\xi_i, \zeta_i$ we can differentiate both sides of equation \eqref{Aux1}, obtaining $\partial_t F_i(t,s) = \partial_t \widetilde{F}_i(t,s)$, that is
	$$
	\iu \langle e^{-\iu t \X} \, \X \, \xi_i ,  e^{\iu s \K} \zeta_i \rangle = \iu e^{-\iu ts } \langle e^{-\iu s \K}  \xi_i,  e^{\iu t \X} \X \zeta_i \rangle - \iu s e^{-\iu ts } \langle e^{-\iu s \K}  \xi_i,  e^{\iu t \X} \zeta_i \rangle \, .
	$$
	Applying again Stone's Theorem, we differentiate in $s$, getting
	\begin{equation*}
	\begin{aligned}
	- \langle e^{-\iu t \X} \,\X\, \xi_i ,  e^{\iu s \K} \,\K \,\zeta_i \rangle =& - e^{-\iu ts } \langle e^{-\iu s \K} \, \K\,  \xi_i,  e^{\iu t \X}\, \X\, \zeta_i \rangle + t e^{-\iu ts } \langle e^{-\iu s \K} \, \xi_i,  e^{\iu t \X} \,\X\, \zeta_i \rangle  \\
	& -\iu e^{-\iu ts } \langle e^{-\iu s \K}  \xi_i,  e^{\iu t \X} \zeta_i \rangle  - s t e^{-\iu ts } \langle e^{-\iu s \K}  \xi_i,  e^{\iu t \X} \zeta_i \rangle\\
	& +  s e^{-\iu ts } \langle e^{-\iu s \K} \K  \xi_i,  e^{\iu t \X} \zeta_i \rangle \, .
	\end{aligned}
	\end{equation*}
	Hence we get $\partial_s \partial_t F_i(t,s) = \partial_s \partial_t \widetilde{F}_i(t,s)$. Evaluating the derivatives at the point $(t,s)=(0,0)$ we obtain
	$$
	 \langle \X\, \xi_i ,  \K \,\zeta_i \rangle =  \langle\, \K\,  \xi_i,  \, \X\, \zeta_i \rangle + \iu  \langle   \xi_i,  \zeta_i \rangle \, .
	$$
	Performing the limit $i \to +\infty$ and taking into account \eqref{ApproxD}, one concludes the proof.
\end{proof}

\begin{thm}[Balian \cite{Bal81}, Low\cite{Low89}, Battle \cite{Bat88}]
	\label{BalianLowTheorem}
	Given a Gabor triple $\Gabor=\left(\X,\K,\Dd\right)$ and a lattice $\Gamma \subset \R^2$, consider a GGF for $\mathcal{V}$ generated by $\varphi_0 \in \mathcal{V}$. Moreover suppose that $\mathcal{V}$ is an invariant subspace for the operators $\X$ and $\K$, that is for every $\varphi \in \mathcal{V}$
	$$
	e^{\iu s \X} \varphi \in \mathcal{V} \, , \qquad  e^{\iu s \K} \varphi \in \mathcal{V} \,.
	$$  
	If the elements of the GGF form a complete and orthonormal system for $\mathcal{V}$ then either $\varphi_0 \notin \Dom(\X)$ or $\varphi_0 \notin  \Dom(\K)$. 
\end{thm}

The latter claim is usually written, especially in the physics literature, as
\begin{equation}
\label{ThBL}
\|\X\varphi_0 \| \|\K\varphi_0 \| = + \infty \, . 
\end{equation}

\noindent Notice, however, that the previous equation makes no sense in an abstract setting. To justify the appearance of \virg{$+\infty$} in \eqref{ThBL}, we notice that - whenever the thesis  of Theorem \ref{BalianLowTheorem} holds true - there exist two sequences of vectors $\{\varphi_i \}_{i\in \N} \subset \Dom(\X)$, with $\varphi_i \to \varphi_0$ as $i \to \infty$, and $\{ \widetilde{\varphi}_j \}_{j\in \N} \subset \Dom(\K)$, with $\widetilde{\varphi}_j \to \varphi_0$ as $j \to \infty$, such that either the sequence $\{ \|\X \varphi_i\| \}_{i\in \N}$ is unbounded or the sequence $\{ \|\K \widetilde{\varphi}_j\| \}_{j\in \N}$ is unbounded.

\bigskip

As we anticipated in the Introduction, we provide a slight generalization of Battle's proof of Balian-Low Theorem. Our proof is representation-independent, in the sense that does not exploit the Stone-von Neumann uniqueness theorem. 
\begin{proof}
	We prove the theorem via reductio ad absurdum. By contradiction, suppose that there exists $c< +\infty$ such that $\|\X\varphi_0 \| \|\K \varphi_0\| = c$ (this means, equivalently, that $\varphi_0 \in \Dom(\X) \cap \Dom(\K)$).  From the invariance property of $\mathcal{V}$ follows that
	$$
	\left(e^{\iu t \X} \varphi_0 - \varphi_0\right) \in \mathcal{V}   \, , \qquad  \left(e^{\iu s \K} \varphi_0 - \varphi_0\right) \in \mathcal{V}  \, .
	$$
	Hence $\X \varphi_0, \K\varphi_0 \in \mathcal{V}$ and we have that 
	\begin{equation}
	\label{starting}
	\begin{aligned}
	\langle \X \varphi_0, \K \varphi_0 \rangle &= \sum_{m,n \in \Z} \langle \X \varphi_0, T_{m,n}\varphi_0 \rangle \langle T_{m,n}\varphi_0, \K \varphi_0 \rangle\, . \\
	\end{aligned}
	\end{equation}
	Let us now prove that $\langle \X \varphi_0, T_{m,n}\varphi_0 \rangle \langle  T_{m,n}  \varphi_0, \K \varphi_0 \rangle=\langle \K \varphi_0, T_{-m,-n}\varphi_0 \rangle \langle  T_{-m,-n}  \varphi_0, \X \varphi_0 \rangle$. By \eqref{Weylrelation} we have that
	\begin{equation}
	\label{Aux2}
	e^{\iu s \X} T_{m,n} e^{-\iu s\X} = e^{-\iu s \left( m a_2 + n b_2 \right)} T_{m,n} \, .
	\end{equation}
	Moreover, since $T$ restricted to $\Gamma$ is a projective unitary representation of $\Z^2$, see Proposition \ref{PropT} , we have that
	\begin{equation}
	\label{Aux3}
	T^*_{m,n}= e^{\iu \left(m a_1 n b_2- m a_2 n b_1\right)} T_{-m,-n} \, .
	\end{equation}
	Putting together \eqref{Aux2} and \eqref{Aux3} we obtain
	\begin{equation}
	\label{Aux4}
	\langle e^{-\iu s \X} \varphi_0 , T_{m,n} \varphi_0 \rangle = e^{-\iu s \left( m a_2 + n b_2 \right)} e^{-\iu \left(m a_1 n b_2- m a_2 n b_1\right)} 	\langle T_{-m,-n} \varphi_0 , e^{\iu s \X } \varphi_0 \rangle \, .
	\end{equation}
	Using the same strategy as in the proof of Proposition \ref{HeisenbergCR}, we differentiate by $s$ both sides of \eqref{Aux4}. Evaluating at $s=0$ we get
	$$
	\begin{aligned}
	 \langle \X \varphi_0 , T_{m,n} \varphi_0 \rangle &=   e^{-\iu \left(m a_1 n b_2- m a_2 n b_1\right)}	\langle T_{-m,-n} \varphi_0 , \X \varphi_0 \rangle\\
	 &\phantom{=} - \left( m a_2 + n b_2 \right) e^{-\iu \left(m a_1 n b_2- m a_2 n b_1\right)} \langle T_{-m,-n} \varphi_0 , \varphi_0 \rangle\, .
	\end{aligned}
	$$
	Since $T_{mn} \varphi_0 \perp \varphi_0$, for every $(m,n)\neq(0,0)$, we obtain
	$$
	\langle \X \varphi_0 , T_{m,n} \varphi_0 \rangle =   e^{-\iu \left(m a_1 n b_2- m a_2 n b_1\right)}	\langle T_{-m,-n} \varphi_0 , \X \varphi_0 \rangle \, .
	$$
	The same argument shows also that
	$$
	\langle T_{m,n} \varphi_0 , \K \varphi_0 \rangle =   e^{\iu \left(m a_1 n b_2- m a_2 n b_1\right)}	\langle  \K \varphi_0 , T_{-m,-n}\varphi_0 \rangle \, .
	$$
	Therefore, by phase cancellation, we conclude from \eqref{starting} that 
	\begin{equation*}
	\begin{aligned}
	\langle\X \varphi_0, \K \varphi_0 \rangle &= \sum_{m,n \in \Z} \langle T_{-m,-n}\varphi_0, \X \varphi_0 \rangle \langle \K \varphi_0, T_{-m,-n} \varphi_0 \rangle\\
	&= \langle\K \varphi_0, \X \varphi_0 \rangle \, .
	\end{aligned}
	\end{equation*}
	Consider now a sequence $\xi_i$ that satisfies item (i) of Assumption \ref{Assumption} with $\psi=\varphi_0$. By using Proposition \ref{HeisenbergCR}, we have
	\begin{equation*}
	\langle \X \xi_i, \K \xi_i \rangle - \langle  \K \xi_i, \X \xi_i \rangle = \iu \|\xi_i\|^2 \, .
	\end{equation*}
	Hence, for $i \to \infty$, it happens that
	\begin{equation*}
	\begin{aligned}
	&\langle \X \xi_i, \K \xi_i \rangle - \langle  \K \xi_i, \X \xi_i \rangle \to \iu \|\varphi_0\|^2 \, ,\\
	&\langle \X \xi_i, \K \xi_i \rangle - \langle  \K \xi_i, \X \xi_i \rangle \to \langle\X \varphi_0, \K \varphi_0 \rangle  -  \langle\K \varphi_0, \X \varphi_0 \rangle = 0 \, .
	\end{aligned}
	\end{equation*}
	This implies that $\varphi_0 = 0$ and so $\{\varphi_{mn}\}_{m,n \in \Z}$ can not be a complete and orthonormal system for $\mathcal{V}$. Thus we get a contradiction and the theorem is proved.
\end{proof}

\section{Application to the Landau model: non existence of well-localized Gabor frames}
\label{Sec:application}
The Balian-Low theorem has been already applied to the Landau Hamiltonian by Zak, \cite{Zak97,Zak98}. These works are based on the theory of linear canonical transformations of Moschinsky and Quesne \cite{MoQu71}, which requires to handle several integral transformations, and on the use of the Bloch--Floquet--Zak transform, also called Zak transform \cite{Zak67,Zak68}. In the following we provide an alternative argument - hopefully more transparent for some readers - based on a direct application of the Balian-Low theorem (Theorem \ref{BalianLowTheorem}) to the Landau levels. Our alternative argument uses only the theory explained so far.

By definition, $\OC_1$, $\OC_2$ and $C_0^{\infty}(\R^2)$ satisfy Assumption \ref{Assumption}.(i). Moreover by explicit computation we get that for every $t,s \in \R$
$$
e^{\iu t \OC_1} e^{\iu s \OC_2} = e^{\iu \frac{ts}{b}}e^{\iu s \OC_2} e^{\iu t \OC_1} \, .
$$
By setting $\widetilde{\OC}_1=-b \OC_1$, we obtain
\begin{equation}
\label{MagWeyl}
e^{\iu t \widetilde{\OC}_1} e^{\iu s \OC_2} = e^{-\iu st}e^{\iu s \OC_2} e^{\iu t \widetilde{\OC}_1} \, .
\end{equation}
Hence $\widetilde{\OC}_1$ and $\OC_2$ together with the dense set $\Dd=C_0^{\infty}(\R^2)$ define a Gabor triple $\Gabor_L$. Consider now the $n^{th}$ Landau level and the lattice $\Z^2$. Since $\OC_1$ and $\OC_2$ commute with the Landau Hamiltonian, in view of \eqref{QDef}, every vector $\varphi$ in the $n^{th}$ Landau level generates a generalized Gabor frame given by
$$
T_{m,n} \varphi = e^{\iu m \widetilde{\OC}_1} e^{\iu n \OC_2} \varphi \, .
$$

Assume that there exists a vector $\varphi_0$ that generates a Generalized Gabor frame that is an orthonormal basis for the $n^{th}$ Landau level.  Applying Theorem \ref{BalianLowTheorem} we conclude that $\varphi_0$ cannot be in both the domain of $\widetilde{\OC}_1$ and of $\OC_2$. Assume now that $\varphi_0$ is not a null vector and is in both the domain of $X_1$ and of $X_2$, or in other words that
\begin{equation}
\label{Aux5}
\|X_1 \varphi_0\| \|X_2 \varphi_0\| < \infty\, . 
\end{equation}
From definition \eqref{QDef}, it follows that
$$
X_1=-\frac{1}{b}\left(\widetilde{\OC}_1+\KM_2\right) \, , \qquad X_2=\OC_2+\KM_1 \, .
$$
Since $\varphi_0$ is an eigenvector of the Hamiltonian $H_L$, we know that $\varphi_0$ is in the domain of $\KM_1$ and $\KM_2$, therefore by linearity it has to be also in the domain of $\OC_1$ and $\OC_2$ in contradiction with Theorem \ref{BalianLowTheorem}. Thus, we conclude that $\varphi_0$ cannot satisfy \eqref{Aux5}, \ie it cannot be well-localized in both directions.

Note that we are not addressing the issue of the existence of a Generalized Gabor frame for a single Landau Level. As it is widely discussed in the literature, see for example \cite{Perelomov1971,BaBuGiKl1971,BoZa78}, the existence of a complete  Gabor frame for a Landau level is related to the existence of a complete von Neumann set for the same space, 
which in turn is related to the choice of the lattice $\Gamma$. The orthogonality and the completeness 
of the elements of the GGF is crucially related to the properties of both the lattice $\Gamma$ and the generator $\varphi_0$ \cite{Simon15}, and can be investigated by using the Zak transform as done in \cite{BaGrZa75, Zak97, Zak98}. Our Theorem says that whenever one can construct such an orthonormal basis for the subspace $\mathcal{V}$, then the elements of the basis cannot be well-localized in position space, in agreement with Zak's result \cite{Zak97,Zak98}, Thouless argument \cite{Thouless84} and the more recent model-independent analysis in \cite{MoPaPiTe1,MoPaPiTe2}.



\bigskip \bigskip


{\footnotesize  

	\begin{tabular}{ll}
	
		(M. Moscolari) 
		&  \textsc{Dipartimento di Matematica, \virg{La Sapienza} Universit\`{a} di Roma} \\
		&  Piazzale Aldo Moro 2, 00185 Rome, Italy \\
		&  {E-mail address}: 
		\href{mailto:moscolari@mat.uniroma1.it}{\texttt{moscolari@mat.uniroma1.it}} \\
		\\
		(G. Panati) 
		&  \textsc{Dipartimento di Matematica, \virg{La Sapienza} Universit\`{a} di Roma} \\
		&  Piazzale Aldo Moro 2, 00185 Rome, Italy \\
		&  {E-mail address}: \href{mailto:panati@mat.uniroma1.it}{\texttt{panati@mat.uniroma1.it}} \\
		\\

	\end{tabular}
	
}

\end{document}